\documentclass{tGIS2e}
\usepackage{natbib}
\usepackage{hyperref}
\usepackage{url}
\usepackage{titlesec}
\usepackage{graphics}

\newcommand{\locx}{\boldsymbol{x}}
\newcommand{\loch}{\boldsymbol{h}}
\newcommand{\vecx}{\boldsymbol{x}}

\newcommand{\vech}{\boldsymbol{h}}
\newcommand{\veczero}{\boldsymbol{0}}
\newcommand{\vecone}{\boldsymbol{1}}

\titlespacing*{\section}{0pt}{1.1\baselineskip}{\baselineskip}
\titlespacing*{\subsection}{0pt}{1.1\baselineskip}{\baselineskip}

\usepackage{amssymb, amsmath, amsfonts, latexsym, makeidx, verbatim, xspace,textcomp}
\usepackage{epsfig, graphicx, graphics}
\graphicspath{{./tp/}}

\begin{document}
\title{On Spatial Transition Probabilities as Continuity Measures in Categorical Fields }
\author{Guofeng Cao$^{a,}$$^{\ast}$\thanks{$^\ast$ Corresponding author.
		Email: guofeng.cao@ttu.edu \vspace{6pt}} and Phaedon C. Kyriakidis$^{b,}$$^{c}$ and Michael F. Goodchild $^{b}$\\\vspace{6pt}
$^{a}$\textit{Department of Geosciences, Texas Tech University, USA} \\
$^{b}$\textit{Department of Geography, University of California, Santa Barbara,
USA} \\ 
$^{c}$\textit{University of the Aegean, Mytilene, Greece}}

\received{Received August 2011; final version received September 2011}
\maketitle
\begin{abstract}
Models of spatial transition probabilities, or equivalently, transiogram models
have been recently proposed as spatial continuity measures in categorical
fields. In this paper, properties of transiogram models are examined
analytically, and three important findings are reported. Firstly, connections
between the behaviors of auto-transiogram models near the origin and the spatial
distribution of the corresponding category are carefully investigated. Secondly,
it is demonstrated that for the indicators of excursion sets of Gaussian random
fields, most of the commonly used basic mathematical forms of covariogram models
are not eligible for transiograms in most cases; an exception is the exponential
distance-decay function and models that are constructed from it. Finally, a
kernel regression method is proposed for efficient, non-parametric joint
modeling of auto- and cross-transiograms, which is particularly useful for
situations where the number of categories is large.

%given a category in a map, analytical connections between the curve of
%auto-transiogram model of this category and the perimeter-area ratio of
%geographic regions in this map with the same category are investigated.  

%In this paper, some results on properties and modeling of spatial transition
%probabilities or transiogram models, a set of recently proposed alternative
%spatial continuity measures in categorical fields, are reported. 
\end{abstract}

\vspace{0.3cm} \noindent {\bf Keywords:} categorical data, transition
probability, geostatistics, spatial continuity

\section{Introduction}
%This paper investigates the properties of transiograms, and presents the
%analytical connections of transiograms with compactness measures (perimeter/area
%ratio) of a geographical shape. In the context of transition probability-based
%indicator geostatistics, the validity of basic transiogram models is analyzed. A
%kernel regression model is also proposed to address the transiogram modeling
%issues in a Markov Chain Random Field (MCRF).

Categorical spatial data, such as land use and land cover data in geography and
environmental science, rock (lithology) types in earth science and
socio-economic survey data in social sciences, etc., are all-important
information sources across a wide spectrum of scientific fields. As with continuous
spatial data, complex spatial patterns (spatial correlation) exist in such
geo-referenced categorical data in accordance with Tobler's first law of
geography \citep{Tobler1970}. A successful investigation of the statistical
characteristics in these patterns will benefit many of the above mentioned
scientific fields, particularly with respect to spatial data classification or
clustering, spatial uncertainty modeling and spatial scale effects. In
remote sensing imagery classification, for example, spatial pattern information
implied in thematic classes (e.g., forest area is more likely adjacent to
grassland than desert area) can be fully integrated with conventional
classifiers to enhance the classification performance \citep{Tso2001}. 

One of the most fundamental concepts in spatial analysis is the choice of
spatial continuity measures usually quantifying similarity in attribute values
or class labels. In conventional geostatistics, indicator kriging (IK)
\citep{Solow1986} and indicator coKriging (ICK) \citep{Deutsch1998} are the most
frequently used methods for estimating the posterior (conditional) probability
of class occurrence at any unsampled location conditioned on the available
observed data. Both IK and ICK rely on two-point spatial continuity measures,
indicator (cross)covariance or (cross)variogram models, for characterizing
spatial association in categorical spatial data. Although covariances and
variograms are suitable for continuous fields, particularly Gaussian random
fields, the discrete characteristics of categorical data, along with their sharp
boundaries and complex spatial patterns, render the interpretation of such
covariances and variograms less intuitive. This, in turn, hinders the
applications of the kriging family of methods for deriving probabilities of
class occurrence in categorical fields. 

Recently, promising alternatives to the indicator (cross-)variogram, namely,
spatial transition probabilities \citep{Carle1996}, or equivalently,
transiograms \citep{Li2006a}, have been proposed as alternative spatial
continuity measures in categorical fields. The concept of transition probability
is not new; but it has only recently been proposed as a continuity measure in
categorical fields.
%The advantages of transiograms over indicator covariances and variogram as
%spatial continuity measure in categorical fields have been intensively studied.
Compared to indicator covariances and variogram models, transiograms are more
interpretable in categorical fields, and easier to integrate with ancillary
information \citep{Carle1996}. Based on this concept, \citet{Carle1996}
reformulated IK and ICK as systems of spatial transition probabilities according
to their analytical connections with indicator covariograms. 
%A multidimensional continuous-lag Markov Chain was proposed to model spatial
%variability in a 3-D geological setting \citep{Carle1997}.  
More recently, \citet{Li2007f} employed a single Markov Chain moving randomly
within a stationary random field (Markov Chain Random Field) for conditional
simulation or interpolation of categorical spatial data. Class occurrence
probabilities derived by such methods usually satisfy the fundamental
probability constraints naturally compared to methods based on variations of IK.

As an extension of transition probabilities in a spatial setting, transiograms
naturally inherit basic properties of two-point conditional probabilities, such
as asymmetry, non-negativity and unit-sum \citep{Carle1996, Carle1997}. 
%On the other hand, they also share the same concepts with indicator variograms,
%such as $range$, $sill$ and $aniostropy$, etc. 
The relationships between the parameters of transiogram models and the
information on class proportion, mean length, and class juxtaposition have been
investigated by \citet{Carle1996}, and this information actually provides an
interpretation of the behavior of transiogram models and eventually offers a
guideline for the construction of such models by incorporating expert knowledge
of the spatial distribution of the categories under study. Along these lines,
one potential contribution of this paper is to investigate the connections
between the behavior of auto-transiogram models near the origin and the spatial
distribution of the associated category.

As with variograms, not every function of distance can serve as a valid
transiogram. Several basic mathematical models of variograms, such as circular,
spherical, exponential, Gaussian, and cosine-Gaussian, have been proposed for
transiogram modeling \citep{Li2006}. No discussion, however, has yet been made
on whether these valid variogram functions can be eligible for transiograms
under certain circumstances. In this paper, the validity of transiogram models
in the stationary indicator random fields, particularly the excursion sets of
Gaussian Random Fields (GRFs), is discussed and it is found that in most cases,
only the exponential form and several of its variants are eligible for
transiogram modeling. 

On another front, even if valid transiograms are assumed to be available, or in
cases where the assumption of stationary indicator random fields does not apply,
e.g., in a Markov Chain Random Field (MCRF) \citep{Li2007f}, transiogram model
fitting from empirical values can become tedious as the number of classes
increases, since there are $K^2$ auto- and cross-transiogram curves to be
jointly modeled for a set of $K$ classes. In practice, one often finds that
basic parametric transiogram models (usually defined by a set of parameters
including range, sill, anisotropy and etc.) cannot capture the irregular
fluctuations at small scales (e.g., hole effect) often found in empirical
transiograms. Incorporating this information in transiogram models could
dramatically increase the number of (unknown) parameters to be estimated
\citep{Li2011}, and renders the quantitative fitting of such models infeasible.
To address these computing and modeling issues, this paper proposes a kernel
regression-based non-parametric fitting procedure for efficient and consistent
transiogram modeling.

The remainder of this paper is organized as follows: the behavior of
auto-transiograms near the origin is examined in Section 3 after briefly
reviewing the concepts of transiograms in Section 2. Section 4 is devoted to
investigating the validity of basic transiogram models in the excursion sets of
GRFs, and the kernel regression based non-parametric transiogram model fitting
method is presented in Section 5. Finally, section 6 concludes the paper and
provides some discussion.

\section{Basic concepts of spatial transition probabilities}

Consider a $d$-dimensional geographical region $\Omega$ which is partitioned into
$N$ disjoint subregions with a categorical random variable (RV) $C(\vecx)$
($\vecx\in R^d$) which can take one out of $K$ mutually exclusive and
collectively exhaustive class labels $c(\vecx)\in \{1,\ldots,K\}$ at any
arbitrary location with coordinate vector $\vecx$. Alternatively, one can also
define an indicator variable $I_k(\vecx)$ to represent $C(\vecx)$, where
$I_k(\vecx)=1$ if $C(\vecx) = k$ and $I_k(\vecx)=0$ otherwise. 

Given two locations $\vecx$ and $\vecx'$ in $\Omega$, and the associated class
labels denoted as $k$ and $k'$, the transiogram $\pi_{k'|k}(\vecx',\vecx)$ is
typically a parametric model of transition probabilities as a function of the
lag vector $\vech=\vecx'-\vecx$. In words, $\pi_{k'|k}(\vecx',\vecx)$ is the
probability of, starting from a source location $\vecx$ with class label $k$,
arriving at a destination location $\vecx'$ with class label $k'$.  Note that
since the dimension of $\Omega$ is greater than $1$,  there are theoretically
infinite paths to reach a destination location $\vecx'$ from a source location
$\vecx$. This equifinality issue hinders the application of the original
definition of 1D transition probability and the rich Markov chain theory based
on it, such as the celebrated Chapman-Kolmogrov equation, to high dimensional
spaces. To eliminate ambiguity, the definition of spatial transition
probabilities is restricted to the path defined by the vector
$\vech=\vecx'-\vecx$. Specifically, spatial transition probabilities could be
defined as:

\begin{eqnarray}
\pi_{k'|k}(\vech) &=&P\{C(\vecx')=k'|C(\vecx)=k\}\nonumber \\
 &=& P\{I_{k'}(\vecx') = 1|I_k(\vecx) = 1\}\nonumber \\
 &=& \frac{P\{I_{k'}(\vecx') = 1 \;and\; I_k(\vecx) = 1\}}{P\{I_k(\vecx) = 1\}}
\label{eq.tp}
\end{eqnarray}

Second-order or intrinsic stationarity \citep{Chiles1999}
is implicitly assumed in this definition, since the value of $\pi_{k'|k}(\vech)$
depends only on the lag $\vech$ and not on the location $\vecx$ or $\vecx'$.
More specifically, $\pi_{k|k}(\vech)$ denotes the {\it auto-transiogram} for class
$k$ (when $k=k'$), a measure of spatial auto-correlation of class $k$, and
$\pi_{k'|k}(\vech)$ denotes the {\it cross-transiogram} from class $k$ to class
$k'$ (when $k\neq k'$), a measure of spatial cross-correlation between class $k$
and class $k'$. Conventionally, class $k$ and class $k'$ in $\pi_{k'|k}(\vech)$
are called {\it tail class} and {\it head class} respectively.  

Under the assumption of second-order stationarity, sample transiograms can be
obtained by direct computation (exhaustive sampling, no parametric model involved) from
sample data on a regular grid whose node spacing coincides with the scale
of analysis. Thus for a given $\vech$, we have: 

\begin{equation}
\hat{\pi}_{k'|k}(\vech) \simeq \frac{1}{\pi_k} E\{I_k(\locx)I_{k'}(\locx+\vech)\}
\simeq \frac{1}{\pi_kN(\vech)} \sum_{n=1}^{N(\vech)}[i_k(\vecx)i_{k'}(\vecx+\vech)]
\label{eq.empirical}
\end{equation}

\noindent where $N(\vech)$ denotes the number of location pairs separated by
vector $\vech$ and $\pi_k$ indicates the proportion of class $k$.
Figure.\ref{fig.sampletp} provides an area-class map with three categories, and
for a certain lag distance $\vech$, the associated transiograms ($3\times 3$)
values are computed by exhaustively enumerating the pairs of nodes separated by a
template vector $\vech$ in the whole sample map.
%scanning the whole map using a template
%vector $\vech$.
%\vspace{.2cm}

\begin{figure}
\hspace*{3cm}\scalebox{0.8}{\includegraphics{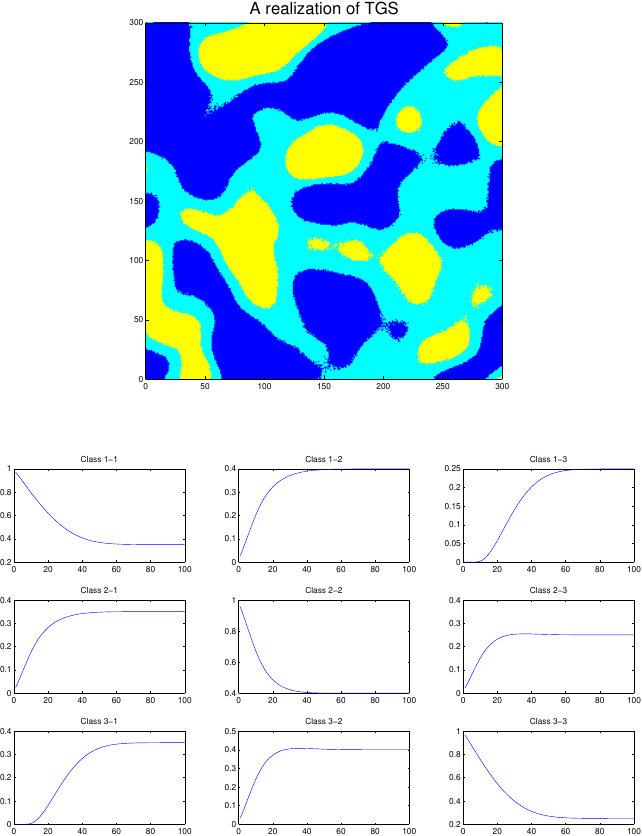}}
\caption{A sample area-class map (top) with three classes and its associated
transiograms curves (bottom) are computed by exhaustive sampling (scanning).
All two-point dependence information of the area-class map is encapsulated in
these ($3\times 3$) empirical transiogram curves.}
\label{fig.sampletp}
\end{figure}

If a (latent) probabilistic distribution is assumed underpinning the
(observable) categorical field, e.g., a truncated multivariate Gaussian field, 
all auto- and cross-transiograms can be computed exactly according to the threshold values
associated with each class and the analytical form of the latent distribution \citep{Chiles1999}.

The basic models of variograms as well as the classical geostatistical concepts
of {\it range}, {\it sill}, {\it hole effect}, and {\it anisotropy}, have been
discussed in the context of transiograms \citep{Li2006a}. 
%It is worthy noting here that the {\it nugget} effect, which is often used to
%account for measurement error and short-range variations in (cross)-variograms,
%is usually discarded in transiograms \citep{Li2007g}.
Given a lag vector $\vech$, transiogram values (spatial transition
probabilities) have the following basic properties:

\begin{itemize}
\item {\it asymmetry}
\begin{equation}
\pi_{k'|k}(\loch)\neq \pi_{k'|k}(-\loch)
\label{eq.asymetry}
\end{equation}
\item {\it non-negativity}
\begin{equation}
\pi_{k'|k}(\loch) \geq 0, \forall k,k' 
\label{eq.nonnegativity}
\end{equation}
\item {\it unit-sum}
\begin{equation}
\sum_{k'=1}^K \pi_{k'|k}(\loch)=1 
\label{eq.sum2one}
\end{equation}
\item {\it value at zero distance}
\begin{equation}
\pi_{k'|k}(0)=\left\{\begin{array}{cc}
1&\mbox{if $k=k'$}\\
0&\mbox{if $k\neq k'$}
\end{array}
\right.
\label{eq.behaviors}
\end{equation}
\end{itemize}

With the basic concepts and properties of the transiogram models reviewed in
this section, the remainder of this paper will investigate the additional
important properties of these models including the properties near the origin,
validity of the transiogram models as well as the fitting procedures of these
models.

%Based on the previous discussion, given two locations $\vecx$ and $\vecx'$, the
%transiogram $\pi_{k'|k}(c(\vecx),c(\vecx'))$ is defined as the spatial
%transition probability function of lag vector $\vech=\vecx'-\vecx$
%(Equation.\ref{eq.tp}), $\pi_{k'|k}(\loch;\boldsymbol{\theta}_{kk'})$, where
%$\boldsymbol{\theta}_{kk'}$ is a parameter vector specific to the pair of
%classes $k$ and $k'$ \citep{Li2006} and is dropped in what follows for
%notational simplicity. 

\section{Behaviors of Transiogram Models Near the Origin}

As in covariograms, the shape (e.g., regularity) of transiograms reflects the
spatial continuity and interaction of categories. In Figure.\ref{fig.sampletp},
for example, the cross-transiogram curves between class $1$ (represented as
$blue$ in Figure.\ref{fig.sampletp}) and class $3$ (represented as $cyan$ in
Figure.\ref{fig.sampletp}) stay at $0$ for a certain distance before they begin
to increase gradually. This transiogram behavior is because class $1$ is never
adjacent to class $3$ in the reference map of Figure.\ref{fig.sampletp}, and
transiogram values should increase after the minimum distance between pixels of
class $1$ and class $3$. This connection between transiogram curves and the
shape of objects, particularly the first derivative of an auto-transiogram curve
at the origin and the shape of objects of the category associated with that
auto-transiogram, is examined in detail in this section.

We start with two particular properties of the indicator representation of
categorical spatial variables:

\begin{equation}
P\{I_k(\locx) = 1\} = E\{I_k(\locx)\}%=\int_{\vecx\in\mathcal{R}^d}i_k(\vecx)d\vecx
\label{eq.I}
\end{equation}
and
\begin{equation}
P\{I_k(\locx) = 1 \; and\;
I_{k'}(\locx+\vech)=1\}=E\{I_k(\locx)I_{k'}(\locx+\vech)\}%=\int_{\vecx\in\mathcal{R}^d}i_k(\vecx)i_{k'}(\vecx+\vech)d\vecx
\label{eq.II}
\end{equation}
\noindent where $E\{I_k(\locx)I_{k'}(\locx+\vech)\}$ is known as the
non-centered indicator cross-covariance or probabilistic version of geometric
covariogram \citep{Lantuejoul2002}. 

Suppose there is a circular region $A$ with category $k$ shifted by $\vech$ to
region $A_h$ as illustrated in Figure.\ref{fig.transition}. In this case, the
area of intersection $A\bigcap A_h$ (blue domain) represents
$E\{I_k(\locx)I_{k}(\locx+\vech)\}$, and the area of $A$ (blue domain and red
domain) represents $P\{I_k(\locx)= 1\}$ if the area of the whole region is assumed
to be $one$. Thus according to the definition of the transiogram
(Equation.\ref{eq.tp}), $\pi_{k|k}(h)$ can be seen as the area of the blue
domain divided by the area of the whole circle or the proportion of the blue
domain in the white or red domain.
\begin{figure}
\begin{center}
\scalebox{0.6}{\includegraphics{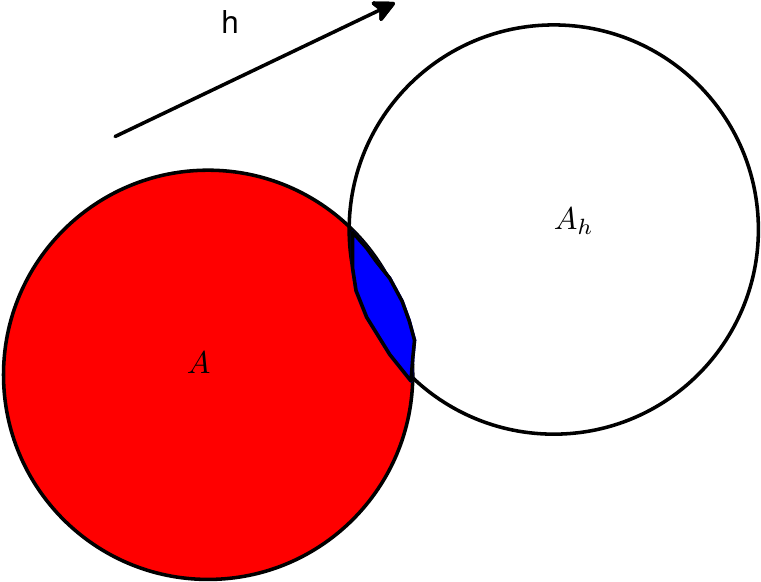}}
\caption{An illustration of non-center auto-transiogram}
\label{fig.transition}
\end{center}
\end{figure}
If the class proportions $E\{I_k(\locx)\}$ are assumed to be constant, one
arrives to the analytical links between the transiogram and the
indicator (cross-)covariogram and indicator (cross-)variogram by applying the
properties of indicators (Equation.\ref{eq.I} and Equation.\ref{eq.II}) to the
definition of the (cross-)covariogram (Equation.\ref{eq.tpcov}) and
(cross-)variogram (Equation.\ref{eq.tpcrossvar}) respectively \citep{Carle1996}: 

\begin{equation}
\sigma_{kk'}(\vech) = \pi_k[\pi_{k'|k}(\vech) -\pi_{k'}]
\label{eq.tpcov}
\end{equation}
\begin{equation}
\gamma_{kk'}(\vech) = \pi_k\{\pi_{k'|k}(\veczero) - [\pi_{k'|k}(\vech) +
\pi_{k'|k}(-\vech)]/2\} 
\label{eq.tpcrossvar} 
\end{equation} 

\noindent where $\sigma_{kk'}(\vech)$ represents indicator (cross-)covariogram
and $\gamma_{kk'}(\vech)$ represents indicator (cross-)variogram and
$\pi_{k'|k}(-\vech)$ represents the transition probability in the opposite
direction of $\vech$. Particularly if we let $k=k'$ and $\pi_{k|k}(0)=1$, 
we have a simple linear connection between the auto-transiogram and the indicator
auto-variogram: 
%\begin{equation}
%\pi_{k|k}(\vech)=\frac{\sigma_{kk}(\vech)}{\pi_k}+\pi_k
%\label{eq.tpcov2} 
%\end{equation}
\begin{equation}
\pi_{k|k}(\vech)=1-\frac{\gamma_{kk}(\vech)}{\pi_k}
\label{eq.tpvar} 
\end{equation}

Because of the linear connection between the transiograms and the indicator
(cross-)variogram/covariogram (Equation.\ref{eq.tpcov} and
Equation.\ref{eq.tpvar}), transiograms share the properties of indicator
(cross-)variogram/covariograms. As Figure.\ref{fig.sampletp} illustrates,
auto-transiograms (diagrams on the diagonal in Figure.\ref{fig.sampletp}) start
from $1$ at $\vech=0$ and gradually decrease to the $sill$ value at $\vech$
equals $range$, i.e., $\lim_{\vech\rightarrow \infty} \pi_{k|k}(\vech)= \pi_k$.
Cross-transiograms (diagrams off the diagonal in Figure.\ref{fig.sampletp}) start
from $0$ at $\vech=0$ and gradually increase to the $sill$ value at $\vech$
equals $range$, i.e., $\lim_{\vech\rightarrow \infty} \pi_{k'|k}(\vech)=
\pi_{k'}$. 

The compactness of a geographic shape is an important property of a polygon in
GIS/zoning and landscape metrics. One of the simplest compactness measures or
indices of a shape is the ratio of its perimeter to its area (perimeter-to-area
ratio) $\Psi$ \citep{Smith2007}. It is well known in the literature that the
first derivative of the variogram at the origin is related to the derivative or
gradient of the surface it represents \citep{Stein1999,Chiles1999}.
\citet{Carle1996} have shown that in a 1D continuous-space Markov chain, the
first derivative of the auto-transiogram at the origin,
$\pi'_{k|k}(\veczero;\phi)$, termed {\it transition rate}, is related to the
{\it mean length} or {\it mean thickness} of the objects of category $k$ in
direction $\phi$. The empirical transition rate is usually calculated by the
total length of category $k$ in direction $\phi$ divided by the number of
embedded occurrences of $k$. In this paper, the relationship between the
auto-transiogram for a certain class label and the perimeter-to-area ratio of
shapes with such a class label in a 2D geographical space is given via the
following proposition.

\begin{proposition}
Under a stationary proportions assumption, the perimeter-to-area ratio $\Psi_k$ of the
objects of category $k$ in a 2D random sets with $unit$ area can be obtained by integrating the derivative
$\pi_{k|k}(0)$ in the direction $\phi$ at the origin over all possible directions:
\begin{equation}
\Psi_k=-\frac{1}{2}\int_{0}^{2\pi}\pi'_{k|k}(\veczero;\phi)d\phi
\label{eq.prop} 
\end{equation}
\label{prop.1}
\end{proposition}
\begin{proof}
In a 2D space, the perimeter $l_k$ of objects of category $k$ can be obtained by the
application of {\it Minkowski's formula} \citep{Matheron1971,Lantuejoul2002}:

\begin{equation} l_k=-\frac{1}{2}\int_{0}^{2\pi}K'_{\phi}(0)d\phi \nonumber
\end{equation} \noindent where $K(\cdot)$ is the probabilistic version of
geometric covariogram for category $k$ \citep{Lantuejoul2002}. Replacing the
covariogram with the transiogram according to Equation.\ref{eq.tpcov}, we have:
\begin{equation} l_k=-\frac{\pi_k}{2}\int_{0}^{2\pi}\pi'_{k|k}(0;\phi)d\phi
\nonumber \end{equation} Under the stationary proportions assumption, $\pi_k$ is
proportional to the area of objects of category $k$, and without loss of
generality, one lets the total area equal $one$ and the area of objects of
category $k$ is thus $\pi_k$.  Equation.\ref{eq.prop} is obtained per the
definition of the perimeter-to-area ratio.  
\end{proof}

For the isotropic case, $\Psi_k$ can then be simply written as:
\begin{equation}
\Psi_k=-\pi\times\pi'_{k|k}(\veczero)
\label{eq.psi}
\end{equation}

If $\pi'_{k|k}(\veczero) =\infty$ in particular, it is anticipated that the
boundaries of category $k$ demonstrate fractal properties. A degenerate case
is when $\pi_{k|k}(\vech)$ is parabolic with $\pi'_{k|k}(\veczero) =0$. 

It is worth noting that, in general, $\Psi_k$ is scale dependent and it is a
particular average shape descriptor (of union) of objects of category $k$
instead of a single object; in other words, $\Psi_k$ corresponds to the {\it
mean perimeter-to-area ratio} in landscape metrics \citep{McGarigal1995}. The
conclusion of the Proposition.\ref{prop.1} is for the closed objects. If the union
object is open, the $boundary$ of the studied area will be taken into account.
This note applies to both convex and concave shapes.

Proposition.\ref{prop.1} can be verified by the circular region (the radius of
this circle $R$ is assumed to be $0.25$) in a map with unit area
(Figure.\ref{fig.transition}). The perimeter-to-area ratio $\Psi_A=\frac{2}{R}=8$
and the auto-transiogram for category $k$ can be written as (the area of blue
domain divided by the area of the circle):  

\begin{equation}
\pi_{k|k}(\vech)=
\frac{2}{\pi}\{\arccos(|\frac{\vech}{2R}|)-|\frac{\vech}{2R}|\sqrt{1-|\frac{\vech}{2R}|^2}\}
\mbox{ if } |\vech|\leq R \nonumber
\label{eq.test}
\end{equation}

Taking the first derivative of Equation.\ref{eq.test} with respect to $\vech$
and letting $\vech=0$, one gets $\pi'_{k|k}(0)=-\frac{2}{\pi R}=-\frac{8}{\pi}$
and Equation.\ref{eq.psi} is obviously satisfied.

Proposition.\ref{prop.1} can further be illustrated in
Figure.3. The $red$ curve in Figure.3
(a) represents the empirical auto-transiogram of class $3$ (represented as
$yellow$) in the sample map Figure.3 (b), and $green$ curve
represents that of the sample map Figure.3 (c). The areas of
these two sample maps are exactly the same and the proportions of class $3$
($yellow$ regions) are both about $0.40$. By comparing
Figure.3 (b) and (c), one can easily check that the average
perimeter of $yellow$ regions in (c) is much larger than that of (b). Thus
according to the proposition.\ref{prop.1}, the auto-transiogram of class $3$
(the $red$ curve) in Figure.3 (b) should have a larger first
derivative (slope) at the origin than that (the $green$ curve) of (c), which is
evident by the two curves in Figure.3 (a).

\begin{figure}
	\hspace*{1cm}\scalebox{0.7}{\includegraphics{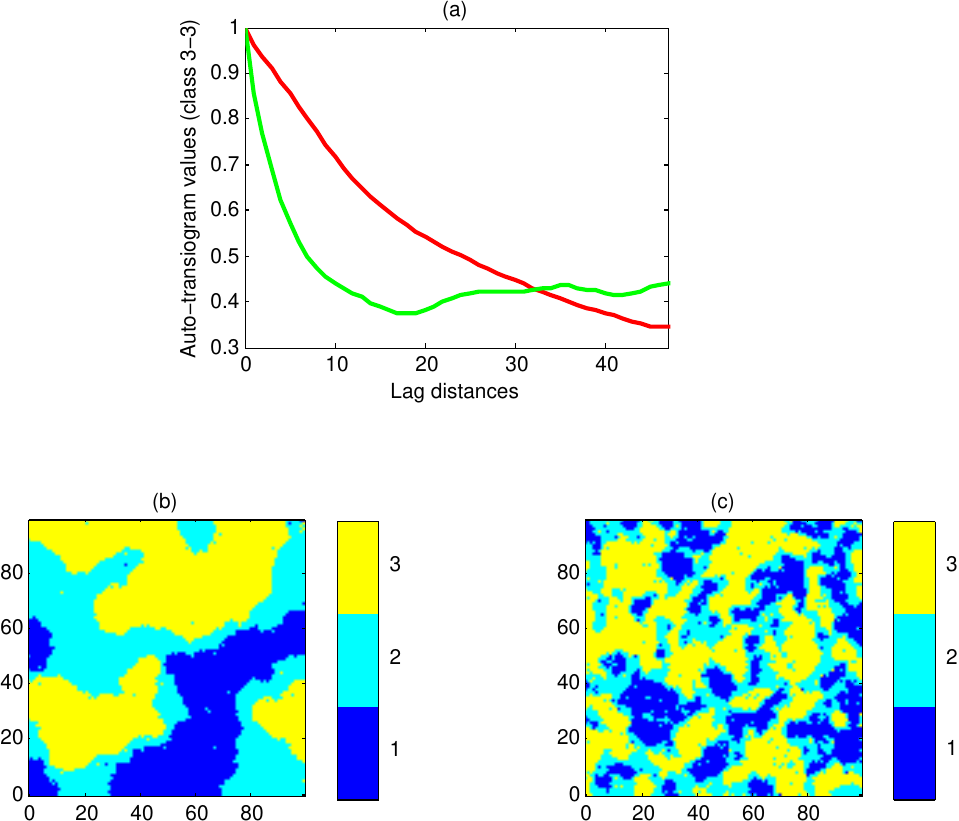}}
\label{fig.prop.sample1}
\caption{An illustration of connections between the first derivative of
auto-transiograms at the origin and the shape of objects. (a) Empirical
auto-transiograms of class $3$ in sample area-class maps (b) and
(c) (b) a sample area-class map with three classes, and the empirical auto-transiogram of class $3$
(represented as $yellow$) is the $red$ curve in (a)
(c) a sample area-class map with three classes, and the empirical auto-transiogram of class $3$
(represented as $yellow$) is the $green$ curve in (a). }
\end{figure}

In this section, the analytical connection between the auto-transiogram of a
certain category, particularly its first derivative at the origin, and the shape
metrics of the objects with this category in a categorical field is carefully
investigated. It provides a quantitative interpretation of the behaviors of
transiogram models. More importantly, together with other properties of
transiograms \citep{Carle1997, Li2006a}, it provides analytical instructions for
incorporating domain experts knowledge in transiogram-based applications.
%\section{Transiograms in Transition Probability-based Indicator Geostatistics}
\section{Valid Transiogram Models}

As in Kriging systems, one often needs to fit empirical transiogram values to
certain parametric transiogram models in the transiogram-based methods
\citep{Carle1996, Li2007f}. Several theoretical models of variograms, such as
triangular, circular, spherical, exponential and Gaussian, have been proposed as parametric
models of transiograms \citep{Li2006a,Li2007g} without checking their validity
under certain circumstances. In this section, this validity is investigated in
the stationary indicator random fields, particularly for indicators of excursion
sets of GRFs, which is commonly used in geostatistics and spatial uncertainty modeling.

Let ${\bf Z}=\{Z(\vecx)\} $ be a stationary GRF with a correlogram $\rho(\vech)$. Oftentimes, the indicators at location
$\vecx$ $i_k(\vecx)$ can be obtained by truncating $Z(\vecx)$ by a specification
of a cut-off value $z_k$, i.e.,
\begin{equation}                                                                                                                                                                               
I_k(\vecx)= \left\{\begin{array}{rl}                                                                                                                                                           
  1 &\mbox{if $Z(\vecx)\geq z_k$}\\                                                                                                                    
  0 &\mbox{if otherwise}                                                                                                                                                                       
 \end{array} \right.                                                                                                                                                                           
\label{eq.definition}
\end{equation}    

It is of interest to determine whether the commonly-used triangular, circular,
spherical, exponential and Gaussian variograms can be used to model the
auto-transiogram of $I_k(\vecx)$. 

%Based on the connections between the indicator covariogram and spatial
%transition probabilities (Equation.\ref{eq.tpcov}), one can reformulate the
%indicator Kriging system in terms of transition probabilities to take advantage
%of the transiogram properties \citep{Carle1996}:
%
%\begin{equation}
%i_k(\vecx_n) = \pi_k + \sum_{n'=1}^N w_{n',k}^{TIK}\pi_{k|k}(\vecx_n -
%\vecx_{n'}),\text{} n=1,\ldots,N
%\label{eq.tiks}
%\end{equation}
%
%\noindent where $i_k(\vecx_n)$ denotes the indicator value for class $k$ at
%location $\vecx_n$ and $w_{n',k}^{TIK}$ represent the reformulated kriging weights 
%associated with location $\vecx_{n'}$.
%
%To ensure that a unique solution ($w_{n,k}^{TIK}$ in equation.\ref{eq.tiks})
%exists, the transition probability matrix on the left-hand side
%of the reformulated linear system (Equation.\ref{eq.tiks}) must be permissible
%(non-singular). In \citet{Carle1996}, transition probability values were
%derived from indicator covariograms and the indicator auto- and
%cross-covariogram models were jointly modeled with a permissible model to arrive
%at a valid indicator covariance matrix which led to valid transition
%probability (left-hand side) matrices in indicator kriging \citep{Carle1996}. In
%this section the permissibility issues of the transition probability matrix are
%analyzed in the context of transiograms, through which a permissible (valid)
%transition probability matrix can be directly obtained. 

The triangular inequality for three random variables $Z(\vecx)$, $Z(\vecx+\vech)$
and $Z(\vecx+\vech+\vech')$ in a stationary random field is written as:

\begin{equation}
|Z(\vecx+\vech+\vech')-Z(\vecx)|\leq |Z(\vecx+\vech+\vech')-Z(\vecx+\vech)| +
|Z(\vecx+\vech)-Z(\vecx)|  \nonumber
\end{equation}
For indicator variables in particular, one has:
\begin{equation}
|I_k(\vecx+\vech+\vech')-I_k(\vecx)|\leq
|I_k(\vecx+\vech+\vech')-I_k(\vecx+\vech)| + |I_k(\vecx+\vech)-I_k(\vecx)| \nonumber
\label{eq.triangle}
\end{equation}

Moreover, $|I_k(\vecx+\vech)-I_k(\vecx)|=|I_k(\vecx+\vech)-I_k(\vecx)|^2$, which
leads to a necessary condition for a valid indicator variogram:
\begin{equation}
\gamma_{kk}(\vech+\vech') \leq \gamma_{kk}(\vech)+\gamma_{kk}(\vech')  
\label{eq.trianglegamma}
\end{equation}

A necessary condition for a valid auto-transiogram can thus be obtained by its
analytical connections with an indicator auto-variogram (Equation.\ref{eq.tpvar}):
\begin{equation}
\pi_{k|k}(\vech+\vech') \geq \pi_{k|k}(\vech)+\pi_{k|k}(\vech')-1
\label{eq.triangletp}
\end{equation}

According to the triangular inequality of auto-transiogram
(Equation.\ref{eq.triangletp}), the Gaussian form (Equation.\ref{eq.gaussian})
cannot be a valid auto-transiogram, since if $\vech=\vech'$, by the
Equation.\ref{eq.trianglegamma} and Taylor expansion, we have $4|\vech|^2\leq
2|\vech|^2$ when $|\vech|\rightarrow 0$, which is obviously a contradiction. 

\citet{Matheron1993} provided a more general necessary condition (containing the
triangular inequality) for eligible indicator variograms $\gamma_{kk}(\vecx,\vecx')$: for any set of $m$ ($m\geq2$) points
$\vecx_1,\ldots,\vecx_m$ with category $k$, and values $\epsilon_i \in
\{-1;0;1\}, i = 1,\ldots, m$, such that $\sum_{i=1}^{m}\epsilon_i=1$, the associated variogram values
$\gamma_{kk}(\vecx_i,\vecx_j)$ must satisfy:
\begin{equation} 
\sum_{i=1}^{m}\sum_{j=1}^{m}\epsilon_{i}\epsilon_{j}\gamma_{kk}(\vecx_i,\vecx_j) \leq 0  
\label{eq.matheron}
\end{equation}

Or equivalently, in terms of auto-transiogram values, one has : 
\begin{equation}
\sum_{i=1}^{m}\sum_{j=1}^{m}\epsilon_{i}\epsilon_{j}(1-\pi_{k|k}(\vecx_i,\vecx_j)) \leq 0  
\label{eq.matheron.tp} 
\end{equation}

%
%The auto-covariogram of the indicator $I(\locx)$ can be written as \citep{Chiles1997}:
%\begin{equation}
%\sigma_{kk}(\vech) = \frac{1}{2\pi} \int_{0}^{\rho(\vech)}
%\exp(-\frac{z^2}{1+u})\frac{du}{\sqrt{1-u^2}}
%\end{equation}

It is still an open question whether the necessary condition is also sufficient
for eligible indicator variograms of general random sets.  \citet{Emery2010}
recently pursued this question further by suggesting that the properties of
triangular, circular, and spherical variograms are rather restrictive in two or
three dimensional indicator random fields, and proved that these three
variograms are not valid indicator variograms for excursion sets of stationary
GRFs. Due to the linear connection between indicator variograms and
auto-transiograms, we can hence conclude that the Gaussian
(Equation.\ref{eq.gaussian}), triangular (Equation.\ref{eq.triangular}),
spherical (Equation.\ref{eq.spherical}) and circular
(Equation.\ref{eq.circular}) models cannot be valid basic forms of
auto-transiograms for indicators of excursion sets of stationary GRFs.

\begin{itemize}
\item {\it Gaussian}
\begin{equation}
\pi_{k|k}(\vech) = 1-(1-\pi_k)\{1-\exp(-(\frac{|\vech|}{a})^2)\}
\label{eq.gaussian}
\end{equation}
\item {\it Triangular}
\begin{equation}
\pi_{k|k}(\vech) = \left\{ 
\begin{array}{cr}
1-(1-\pi_k)\frac{|\vech|}{a} & \mbox{ if $\frac{|\vech|}{a} \leq 1$} \\
\pi_k & \mbox{$o.w.$}
\end{array} \right.
\label{eq.triangular}
\end{equation}
\item {\it Spherical}
\begin{equation}
\pi_{k|k}(\vech) = \left\{ 
\begin{array}{cr}
1-(1-\pi_k)\{1.5\frac{|\vech|}{a}-0.5(\frac{|\vech|}{a})^3\} & \mbox{ if $\frac{|\vech|}{a} \leq 1$} \\
\pi_k & \mbox{$o.w.$}
\end{array} \right.
\label{eq.spherical}
\end{equation}
\item {\it Circular}
\begin{equation}
\pi_{k|k}(\vech) = \left\{
\begin{array}{cr}
1-(1-\pi_k)\{1-\frac{2}{\pi}[\arccos(\frac{|\vech|}{a})-\sqrt{1-(\frac{|\vech|}{a})^2}]\} & \mbox{ if $\frac{|\vech|}{a}\leq 1$}\\
\pi_k & \mbox{$o.w.$}
\end{array} \right.
\label{eq.circular}
\end{equation}
\end{itemize}
\noindent where $a$ represents the range parameter of the model.

To illustrate this conclusion, particularly for spherical cases,  we assume that
a spherical form (Equation.\ref{eq.spherical}) with $a=1$ is used to model the
auto-transiogram of indicator variables $I_k(\locx)$ truncated by
Equation.\ref{eq.definition} from a stationary GRF with unit variance and a
correlogram $\rho(\vech)$. The class proportion $\pi_k$ can be written as
$1-\Phi(z_k)$ with $\Phi(\cdot)$ indicates the $cdf$ (cumulative distribution
function) of Gaussian distribution.  This auto-transiogram is equivalent to a
spherical auto-variogram with sill $\Phi(z_k)(1-\Phi(z_k))$ and range $a = 1$.
From another perspective, the auto-variogram of the indicator $I(\locx)$ can be
obtained by a function of $\rho(\vech)$ \citep{Chiles1999}: 
\begin{equation}
\gamma_{kk}(\vech) = \frac{1}{2\pi} \int_{\rho(\vech)}^{1}
\exp(-\frac{z^2}{1+u})\frac{du}{\sqrt{1-u^2}}
\label{eq.gammarho}
\end{equation}
One can thus have the corresponding correlation function $\rho(\vech)$ of the GRF by inverting
Equation.{\ref{eq.gammarho}}. Given certain number of locations with
coordinates, $\rho(\vech)$ could lead to a singular covariance matrix (see proof
of Proposition 14 in \citep{Emery2010}), which shows that the previous
assumption of auto-transiogram is not true and thus verifies the conclusion that
a spherical form can not be used to model the auto-transiogram of excursion sets
of GRFs.

Fortunately, the exponential variogram and its derived models are valid
indicator variograms in any Euclidean space \citep{Emery2010}. Specifically, the
exponential variogram could be written as
$\gamma_{kk}(\vech)=\pi_k(1-\pi_k)(1-\exp(-\frac{|\vech|}{a}))$, where
$\pi_k(1-\pi_k)$ is the variance of indicator variable $I_k(\vecx)$. According
to Equation.\ref{eq.tpvar}, an eligible auto-transiogram thus can be given in
Equation.\ref{eq.exponential}.
%and the associated cross-transiogram can be
%further obtained (Equation.\ref{eq.crossexponential}) by the {\it unit-sum}
%requirement. 
This conclusion is not surprising considering transition
probabilities are written as an exponential function of transition rates in
continuous-time or in 1D continuous-space Markov chain \citep{Carle1997}. The
{\it memoryless} property of the exponential variogram model, which in a spatial
context states that the value of a geo-referenced variable $Z(\vecx)$ depends
only on its local neighbors, provides a foundation to simplify computations by
reducing the global spatial interactions to local. 

\begin{itemize}
\item {\it Exponential auto-transiogram}
\begin{equation}
\pi_{k|k}(\vech) = 1-(1-\pi_k)\{1-\exp(-\frac{|\vech|}{a})\}
\label{eq.exponential}
\end{equation}
%\item {\it Exponential cross-transiogram}
%\begin{equation}
%\pi_{k'|k}(\vech) = \pi_{k'}\{1-\exp(-\frac{|\vech|}{a})\} \mbox{ for $k\neq k'$}
%\label{eq.crossexponential}
%\end{equation}
\end{itemize}

According to the connections between the indicator covariogram and spatial
transition probabilities (Equation.\ref{eq.tpcov}), \citet{Carle1996}
reformulated the indicator Kriging system in terms of transition probabilities
to take advantage of the transiogram properties.
%\begin{equation}
%\pi_{k|k}(\vecx - \vecx_n) = \sum_{n'=1}^N w_{n',k}^{TIK}\pi_{k|k}(\vecx_{n'} - \vecx_{n}),\text{} n=1,\ldots,N
%\label{eq.tiks}
%\end{equation}
%\noindent where $i_k(\vecx_n)$ denotes the indicator value for class $k$ at
%location $\vecx_n$ and $w_{n',k}^{TIK}$ represent the reformulated kriging weights 
%associated with location $\vecx_{n'}$.
An immediate consequence of the discussion in this section is that the
exponential form (Equation.\ref{eq.exponential}) is recommended over the
triangular (Equation.\ref{eq.triangular}), circular
(Equation.\ref{eq.circular}), Gaussian (Equation.\ref{eq.gaussian}) and
spherical (Equation.\ref{eq.spherical}) forms for modeling the transiograms in
the cases of excursion sets of GRFs and eventually building the transition
probability matrix of the reformulated system.
\section{Non-parametric Transiogram Modeling}

Based on the concepts of transiograms, several approaches have been proposed for
categorical spatial data modeling from different perspectives. \citet{Li2007f}
proposed a Markov Chain random field (MCRF) by applying a spatial Markov Chain
in a 2D geographical space. \citet{Allard2011} uses a similar concept, namely the {\it bi-probagram},
which is basically a bivariate joint probability function of lag distances.
\citet{Cao2011}, on the other hand, proposed a redundancy model in categorical
fields to relax the strict conditional independence assumption commonly imposed
in transiogram-based methods. Different from stationary indicator random fields,
such as mosaic random fields, the Boolean random sets and the excursion
set of GRFs discussed in the previous section, transiograms in these models must
only meet basic probability constraints (Equation.\ref{eq.nonnegativity} to
Equation.\ref{eq.sum2one}) and not Matheron's conditions
(Equation.\ref{eq.matheron}). Although this results in more options for valid
transiograms, the joint fitting of transiograms becomes tedious as the number of
classes increases. From another perspective, the shapes of valid transiograms
are usually controlled by a set of parameters including range, sill and
anisotropy values. These basic shapes, however, tend to be over-smoothed and
ignore the small scale effects found in empirical transiogram values. To account
for such effects, new shape parameters and primitives (e.g., trigonometric
functions for periodic effects) are usually imposed on the basic transiogram
models, and oftentimes, this results in dramatic increases in the complexity of
the transiogram models \citep{Li2011} and  makes parameter fitting
computationally infeasible.  In what follows, a Nadaraya-Watson kernel smoothing
regression \citep{Nadaraya1964} based method is proposed for non-parametric
transiogram modeling to address these problems.

\subsection{Kernel Regression for Transiogram Modeling}

Suppose for $\pi_{k'|k}(\vech)$, the transiogram from class $k$ to class $k'$ at
a certain direction, we have empirical transiogram values $p_{k'|k}(\vech_1),
\ldots, p_{k'|k}(\vech_N)$ for lag $\vech_1,\ldots, \vech_N$ respectively. To
compute $\hat{\pi}_{k'|k}(\vech^*)$, the transiogram from class $k$ to
class $k'$ for an arbitrary lag $\vech^*$ in the same direction, one first finds the range 
$[\vech_n,\vech_{n+1}]$ in which $\vech^*$ lies, and then $\hat{\pi}_{k'|k}(\vech^*)$ can be written
as a linear interpolation of empirical transiogram values \citep{Li2010}:
\begin{equation}
\hat{\pi}_{k'|k}(\vech^*)=\frac{p_{k'|k}(\vech_n)|\vech_{n+1}-\vech^*|+p_{k'|k}(\vech_{n+1})|\vech^*-\vech_{n}|}{|\vech_{n+1}-\vech_{n}|}
\label{eq.linear}
\end{equation}
It has been shown that Equation.\ref{eq.linear} meets the required
probability constraints (Equation.\ref{eq.nonnegativity} to
Equation.\ref{eq.sum2one}), but the formulation is rather limited and only
linear effects in transiogram values are taken into account, which is apparently
an unrealistic assumption for real cases.

Using principles of kernel density estimation (Parzen window estimation)
\citep{Silverman1986}, the Nadaraya-Watson kernel smoothing regression
method \citep{Nadaraya1964} has been proposed for non-linear regression, as
non-linear effects can be modeled by carefully selected kernel functions.
This non-parametric kernel techniques have been previously used for
bi-probagram fitting \citep{D2004, Allard2011}. Similarly here, by applying the
Nadaraya-Watson kernel regression, the resulting transiogram model value
$\hat{\pi}_{k'|k}(\vech^*)$, can be given by:
\begin{equation}
\pi_{k'|k}(\vech^*)= \left\{\begin{array}{cc}
\frac{\sum_{n=1}^{N}\kappa(|\vech_n-\vech^*|)p_{k'|k}(\vech_n)}{\sum_{n=1}^{N}\kappa(|\vech_n-\vech^*|)}
& \mbox{if $|\vech^*| \neq 0$}\\
1 & \mbox{if $|\vech^*|=0$ and $k=k'$} \\
0 & \mbox{if $|\vech^*|=0$ and $k\neq k'$} 
\end{array} \right.
\label{eq.kernelregression}
\end{equation} 
where $\kappa(\cdot)$ is a kernel function with bandwidth $r$. Note that
$\hat{\pi}_{k'|k}(\vech^*)$ is not continuous at the origin, and this discontinuity can be
regarded as {\it nugget effect} usually caused by noise and measurement
errors.

As a kernel function, $\kappa(\cdot)$ should satisfy the following requirements:
\begin{itemize}
\item {\it non-negative}
\begin{equation}
\kappa(t) \geq 0; \mbox{, } \forall t\in \mathcal{R}  \nonumber
\label{eq.kernel1}
\end{equation}
\item {\it unit-integral}
\begin{equation}
\int_{-\infty}^{+\infty} \kappa(t) dt = 1 \nonumber
\label{eq.kernel2}
\end{equation}
\item {\it symmetry}
\begin{equation}
\kappa(t) = \kappa(-t) \nonumber
\label{eq.kernel3}
\end{equation}
\end{itemize}

The properties of commonly-used kernel functions (Table.\ref{tab.kernel}),
such as Gaussian, Epanechnikov, Biweight, Triangular, have been studied
extensively \citep{Silverman1986}. Usually the Epanechnikov
kernel tends to generate the smallest square errors if the smoothing parameter
$r$ is chosen correctly. The recently proposed linear interpolation method for
transiogram modeling \citep{Li2010} can be regarded as a special case of
Equation.\ref{eq.kernelregression} if a triangular kernel is selected and for
each $\vech^*$, only the two nearest neighbors are chosen. The Gaussian
kernel, one of the most commonly-used kernel functions, generates the smoothest curves and
thus tends to create smooth category boundaries in the output map. 
Higher values of $r$ (bandwidth of kernel functions), lead to smoother results, and numerous ways have
been proposed to obtain the optimal $r$, including least-squares
cross-validation, likelihood cross-validation, reference to a standard
distribution and subjective choices \citep{Silverman1986}. 

\vspace{.2cm}
\begin{table}
\begin{tabular*}{\textwidth}{c|c}
\hline
Kernel & $\kappa(t)$ where $t=\frac{\Delta h}{r}$\\ \hline
Epanechnikov & $\frac{3}{4}(1-t^2)\vecone_{|t|\le 1}$\\ 
Gaussian & $\frac{1}{\sqrt{2\pi}}\exp\{-t^2\}$\\ 
Biweight & $\frac{15}{16}(1-t^2)^2\vecone_{|t|\le 1}$ \\
Triangular & $(1-|t|)\vecone_{|t|\le 1}$ \\ \hline
\end{tabular*}
\vspace*{.3cm}
\caption{Several typical kernel functions for kernel regression}
\label{tab.kernel}
\end{table}

A proof is given to show that Equation.\ref{eq.kernelregression} always
yields valid transiogram values. 

\begin{proof}
Assume that empirical transiogram values are obtained by exhaustive sampling of
all observed data, thus given a lab $\vech$, these empirical transiograms should
meet basic transiogram constraints, i.e., for a given $\vech_n$, we
have $p_{k'|k}(\vech_n)\geq 0$ and $\sum_{k=1}^{K}p_{k'|k}(\vech_n)=1$.

According to the definition of the transiogram (Equation.\ref{eq.tp})and
calculation of empirical transiogram values $p_{k'|k}(\vech_n)$
(Equation.\ref{eq.empirical}), the values of the transiograms at the origin
(Equation.\ref{eq.behaviors}) is satisfied naturally, i.e., $p_{k'|k}(0)=1$ if
$k=k'$ and $p_{k'|k}(0)=0$ if $k\neq k'$. 

Since we have $\kappa(\cdot)\ge 0$ for a valid kernel function, 
the non-negativity of transiograms is obviously satisfied.

There is only one unknown parameter $r$, for the head class $k$ and a given
$\vech^*$, $\kappa(|\vech_n-\vech^*|)$ will be the same value in all auto- and cross-transiograms. Thus:
\begin{equation}
\sum_{k=1}^{K}p_{k'|k}(\vech^*) =
\sum_{k=1}^{K}\frac{\sum_{n=1}^{N}\kappa(|\vech_n-\vech^*|)p_{k'|k}(\vech_n)}{\sum_{n=1}^{N}\kappa(|\vech_n-\vech^*|)}
= \frac{\sum_{i=1}^{N}\kappa(|\vech_n-\vech^*|)\sum_{k=1}^{K}p_{k'|k}(\vech_n)}{\sum_{n=1}^{N}\kappa(|\vech_n-\vech^*|)}=1
\nonumber
\end{equation} 
\noindent Thus the unit-sum property of transiograms is satisfied.
\end{proof}

As an example with two categories, Figure.\ref{fig.kernel_transiogram} gives a
comparison between the proposed kernel regression method and the linear
interpolation method \citep{Li2010}. Both methods yield valid transiogram
values. Not surprisingly, the linear interpolation method (green solid lines)
only captures the linear effects in transiogram values, and the proposed method
(solid red lines) tends to generate much smoother results by accounting for the
non-linear effects in transiogram values via kernel functions, and the shapes of
the outcome transiograms of the proposed method can be flexibly adjusted through
the parameters of the kernel functions. In contrast to the linear interpolation
method that only works within the range of empirical values, the proposed method
can be used for extrapolation for distances beyond the empirical range.  In
addition, the proposed method is not a exact estimator, which means its output
estimated transiogram values do not always reproduce the empirical transiogram
values (blue circles) as those of the linear interpolation method do. The
discrepancy between the estimated and empirical transiogram values is controlled
by the kernel bandwidth. Particularly, this discrepancy at the origin (when
$|\vech^*|=0$) can actually be taken as the {\it nugget effect} that the linear interpolation
method ignores. To render this proposed non-parametric method operational, the
described procedure has been implemented in Matlab ($ksr2.m$), and integrated with
a toolbox for statistical analysis of categorical spatial data, which is
available at: \url{http://www.geog.ucsb.edu/~cao/research.html}.
\begin{figure}
\scalebox{0.5}{\includegraphics{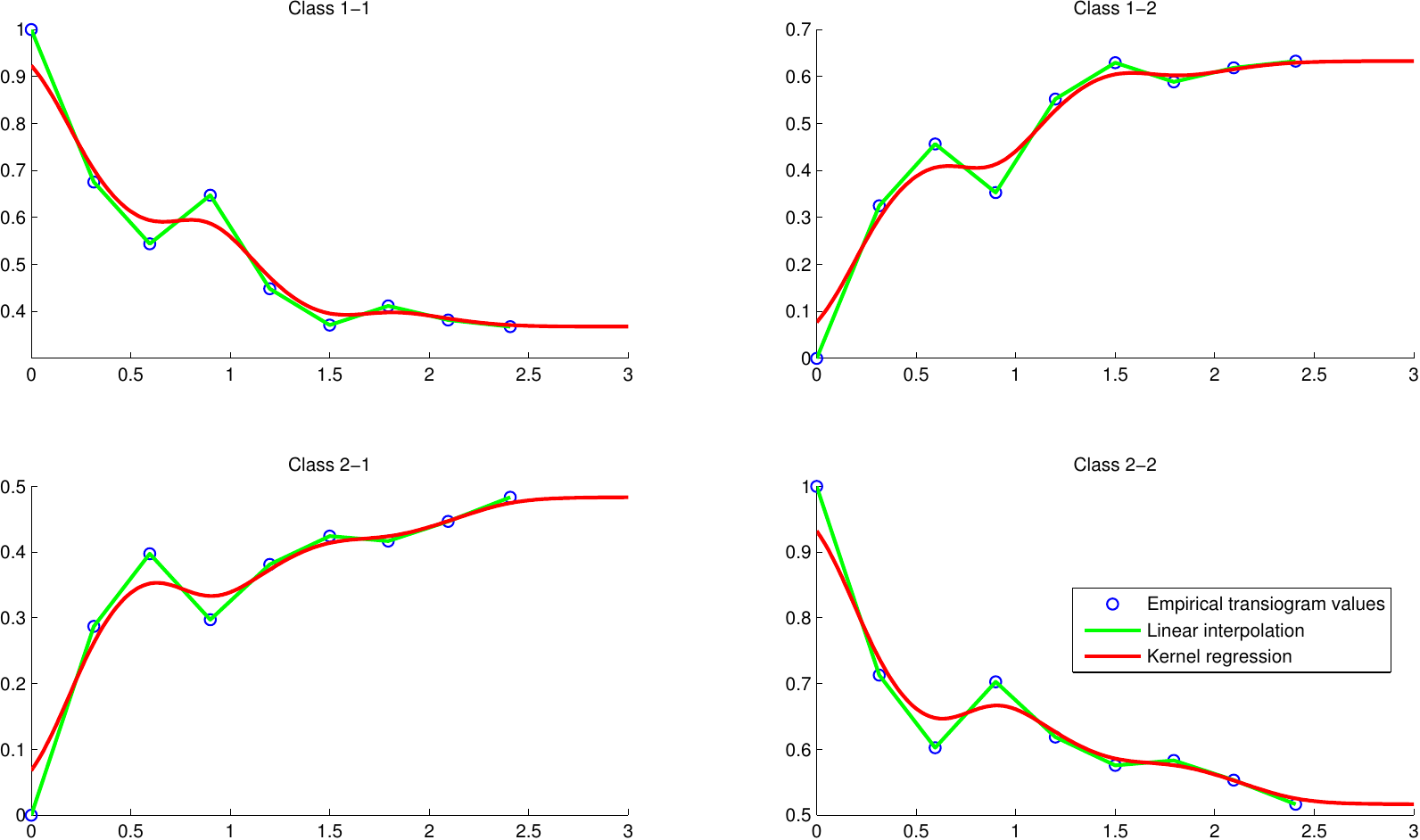}}
\caption{Non-parametric auto- and cross-transiogram models obtained via kernel
regression and linear interpolation}
\label{fig.kernel_transiogram}
\end{figure}

\section{Conclusions}
A collection of statistical methods have been recently proposed for modeling
categorical spatial data based on the concept of spatial transition
probabilities. Limited discussions, however, have given to the properties of this
fairly new spatial continuity measure in the existing literature. In this paper,
three findings on basic properties of transiogram models are reported.  Specifically,
analytical connections between the shape of auto-transiograms near the origin
and the spatial distribution of the associated class label was firstly revealed.
Similar to variograms, it is not every function that can be used as a valid
transiogram model. In the context of stationary indicator random fields, the
eligibility of commonly used basic forms of variograms as transiograms was
investigated particularly for the excursion sets of Gaussian random fields, one
of the most commonly used random sets. It was concluded that the
auto-transiogram of indicators in such a random set can not be Gaussian,
Spherical, Circular or Triangular forms, which are usually used for variogram
modeling. The exponential and its derived forms are recommended for transiogram
modeling in the methods based on stationary indicator random fields. Finally, a
non-parametric transiogram fitting procedure was proposed for the cases where
the assumption of the stationary indicator random fields does not apply, to
capture the small scale effects in transiograms and to address automatic joint
fitting issues of transiograms as the number of classes increases. Compared to
the recent joint fitting methods based on linear interpolation, the proposed
kernel regression-based method is more generic and flexible, and capture the
non-linear effects in empirical transiogram values naturally. A Matlab
implementation of the proposed joint fitting procedure of empirical transiogram
values is also provided. These three findings cover the properties, validity and
modeling of transiograms, and provide a better understanding of the behaviors of
transiogram models, as well as the transiogram-based methods, and thus to avoid
the potential mis-use and misunderstanding of this fairly new spatial continuity
measure in geographical spaces.

\section{Acknowledgments}

We gratefully acknowledge the funding provided by
the National Geospatial-Intelligence Agency (NGA) to support this research.  

\bibliographystyle{tGIS}
%\bibliography{library-no-urls}
\bibliography{main}
\end{document}